\documentclass[sigplan,screen,acmthm,]{acmart}

\usepackage{ifdraft}

\ifdefined\ifdraft\else 
  \makeatletter                   
  \def\mdseries@tt{m}
  \@ACM@printacmreffalse
  \@ACM@nonacmtrue
  \makeatother                    
  \usepackage[plain]{fancyref}
  \usepackage{color}
  \usepackage{hyperref}           
  \hypersetup{
    colorlinks=true,
    linkcolor=blue,
    filecolor=red,      
    urlcolor=magenta,
    breaklinks=true,            
  }
  \usepackage{breakurl}         
\fi

\usepackage[english]{babel}

\usepackage{mathtools,bm}

\usepackage[multiuser,layout=inline]{fixme}
\usepackage{shuffle}

\FXRegisterAuthor{mc}{amc}{MC}
\FXRegisterAuthor{cp}{acp}{CP}
\FXRegisterAuthor{cb}{acb}{CB}
\FXRegisterAuthor{tz}{atz}{TZ}
\fxsetup{theme=color}

\usepackage{cleveref}

\input{define.orgtex}

\begin{document}

\title[The Regular Languages of First-Order Logic with One Alternation]{The
  Regular Languages of First-Order Logic\\with One Alternation}

\author{Corentin Barloy}
\email{corentin.barloy@inria.fr}
\affiliation{%
  \institution{Univ. Lille, CNRS, INRIA, Centrale Lille, UMR 9189 CRIStAL}
  \country{France}
}

\author{Michaël Cadilhac}
\email{michael@cadilhac.name}
\orcid{0000-0001-9828-9129}
\affiliation{%
  \institution{DePaul University}
  \country{USA}}

\author{Charles Paperman}
\email{charles.paperman@univ-lille.fr}
\orcid{0000-0002-6658-5238}
\affiliation{%
  \institution{Univ. Lille, CNRS, INRIA, Centrale Lille, UMR 9189 CRIStAL}
  \country{France}
}

\author{Thomas Zeume}
\email{thomas.zeume@rub.de}
\orcid{0000-0002-5186-7507}
\affiliation{%
  \institution{Ruhr-Universität Bochum}
  \country{Germany}
}

\begin{abstract}
  The regular languages with a neutral letter expressible in first-order logic
  with one alternation are characterized.  Specifically, it is shown that if an
  arbitrary \(\Sigma_2\) formula defines a regular language with a neutral letter, then
  there is an equivalent \(\Sigma_2\) formula that only uses the order predicate.  This
  shows that the so-called Central Conjecture of Straubing holds for
  \(\Sigma_2\) over languages with a neutral letter, the first progress on the
  Conjecture in more than 20 years.  To show the characterization, lower bounds
  against polynomial-size depth-3 Boolean circuits with constant top fan-in are
  developed.  The heart of the combinatorial argument resides in studying how
  positions within a language are determined from one another, a technique of
  independent interest.
\end{abstract}



\keywords{automata theory, first-order logic, descriptive complexity, circuit
  complexity.}

\maketitle

\section{Introduction}

\paragraph{Circuits and regular languages.} Since the works of Barrington and
Thérien~\cite{barrington89,barringtont88}
in the early 1990s, regular languages have emerged as the backbone of
small-depth circuit complexity.  Despite being the most elementary class of
languages, regular languages seem to embody the intrinsic power of circuit
classes.  Under a suitable notion of reduction, a lot of relevant circuit
classes even admit \emph{complete} regular languages (this is at the heart of
Barrington's
Theorem~\cite{barrington89}).
In addition, it seems that each natural restriction of small-depth circuits
defines its \emph{own} class of regular languages. 

More precisely, consider the
following families:
\begin{itemize}
\item \(\AC^0_i\) is the class of Boolean circuits families of depth~\(i\) and polynomial size,
\item \(\ACC^0_i\) is the same as \(\AC^0_i\) but with additional modulo gates, 
\item \(\TC^0_i\) is the same as \(\AC^0_i\) but with additional \emph{threshold}
  gates (more than half of the inputs are 1).
\end{itemize}
The hierarchy in depth of \(\AC^0\) is known to be
strict~\cite{sipser83}, but this is open for the other
classes (for \(\TC^0\) this is known up to depth
3~\cite{hajnalmpst93}).  \mcnote{What about ACC?
  Balaji is unsure.} It is however conjectured that each of these hierarchies is
strict and that strictness can always be witnessed by \emph{regular} languages;
in other words, as mentioned, each of these classes is conjectured to have its
\emph{own} subset of regular languages.

Over the past 30 years, abundant literature has provided a sophisticated
toolset to show separation (and sometimes decidability) of classes of
\emph{regular} languages.  This toolset relies on algebraic objects that
characterize the complexity of regular languages; this object satisfies some
properties iff the language belongs to some class.  It is thus tempting to,
first, characterize the regular languages that belong to a circuit class, and,
second, separate the classes that arise.  This paper is focused on that first
step, for a specific class of circuits (loosely speaking, \(\AC^0_3\)).

\paragraph{Logic.}  Wishing to provide a guiding light, Straubing~\cite{straubing94}
presented in a succinct but beautiful way the links between circuit complexity
and automata theory, and formulated a conjecture on\ldots{} logics\tznote{What is meant by this?}.  Indeed,
circuits themselves are ill-formed for statements of the form ``a circuit family
\(\AC^0_i\) recognizes a regular language iff it has this specific shape.''  Logic
came to the rescue by giving a descriptional tool for circuits.  This started
with the work of Barrington et al.~\cite{barringtoncst92}
and Straubing~\cite{straubing91} who showed that
\(\AC^0\) ($= \bigcup_i \AC^0_i$) is equivalent to first-order logic.  This means that
for any \(\AC^0\) circuit family, there is a first-order formula, with quantifiers
over positions, that recognizes the same language.  For instance, over the alphabet \(A=\{a,
b,c \}\) the language \(A^*ab^*aA^*\), which is in \(\AC^0_2\), can be written as:
\begin{align*}
  (\exists x, y)[&x < y \land a(x) \land a(y) \land \mbox{} \tag{find the 2 \(a\)s}\\
           &(\forall z)[x < z < y \rightarrow b(z)]] \tag{everything in between is a \(b\)}
\end{align*}

In this formula, we used the \emph{numerical predicate} \(<\); numerical
predicates speak about the \emph{numerical value} of positions but not their
contents.  The class of (languages recognized by) formulas \(\FO[\arb]\) is that
of first-order formulas where we allow \emph{any} numerical predicate (even
undecidable ones!).  The aforementioned characterization reads:
\(\AC^0 = \FO[\arb]\).  Extensions of this tight relationship between circuits and
logics exist for \(\ACC^0,\) \(\TC^0,\) and other classes (see~\cite{straubing94}).

Equipped with this, the characterization of the regular languages of \(\AC^0\) is
given by this striking statement~\cite{straubing91}:
\[\FO[\arb] \cap \Reg = \FO[\reg],\]
where \reg is the set of numerical predicates \(<, +1,\) and divisibility by
constants.  Straubing notes \emph{``this phenomenon appears to be quite
  general''} and postulates that for well-behaved logics \(\cL\), it holds that
\[\cL[\arb] \cap \Reg = \cL[\reg].\]
This is known as the Straubing Property for \(\cL\) (simply \emph{Central
  Conjecture} in~\cite{straubing94}) and it is explicitly stated for the logics
\(\Sigma_i\).  Straubing~\cite[p.~169]{straubing94} explains: \emph{``This has the look of a
  very natural principle. It says, in effect, that the only numerical predicates
  we need in a sentence that defines a regular language are themselves
  recognized by finite automata.''}

The Program for Separation in circuit complexity then becomes:
\begin{enumerate}
\item Identify a logic that corresponds to the circuit class,
\item Prove the Straubing Property for the class,
\item Show separation over regular languages.
\end{enumerate}

In this work, we apply this approach to the lower reaches of the \(\AC^0_i\)
hierarchy.  Step 1 in the Program is covered by a result of~\cite{macielpt00}; we
will be focusing on the subset \(\Sigma_i\) of \(\FO\) of formulas with \(i\) quantifier
alternations, starting with an existential one.  For instance, the above
formula is in \(\Sigma_2[<, +1]\).  The Straubing Properties for these logics are very much
open: it is known to hold for \(\Sigma_1\)~\cite{straubing94} and its Boolean
closure~\cite{macielpt00,straubing01}, but no progress has been made on Straubing
Properties since the end of the 1990s.

We will be mostly focusing on languages with a \emph{neutral letter}, this means
that the languages will admit a letter \(c\) that can be added or removed from
words without impacting their membership in the language.  This is usually not a
restriction to the Program for Separation, since it is conjectured that circuit
classes are separated by regular languages with neutral letters.  See, in
particular, the fascinating survey by
Koucký~\cite{koucky09}.  Write \(\NL\) for the set
of languages that have a neutral letter.  The \emph{Neutral} Straubing Property
is that for a logic \(\cL\), it holds that:
\[\cL[\arb] \cap \Reg \cap \NL = \cL[<] \cap \NL.\]

\paragraph{Contributions.}  We show that \(\Sigma_2\) has the Neutral Straubing
Property and that \(\Delta_2 = (\Sigma_2 \cap \Pi_2)\) has the Straubing Property, where
\(\Pi_2\) is defined as \(\Sigma_2\) but with \(\exists\) and \(\forall\) swapped.  Consequently, we
exhibit some natural regular languages that separate \(\AC^0_2\) and \(\AC^0_3\).

\paragraph{Related work.}  In \cite{grosshansms17}, the authors study the model of
so-called programs over DA.  This defines another subclass of $\AC^0$, but it is
not known to have any equivalent characterization in terms of circuits. It is in
particular not known to be equivalent to $\Sigma_2[\arb] \cap \Pi_2[\arb]$ or to the
two-variable fragment of $\FO[\arb]$, classes that we will explore in
\Cref{sec:consequences}.  They give a precise description of the regular
languages computable with programs over DA, but their proof uses the algebraic
structure of DA, which is not available in our setting.

The class \(\Sigma_2[<]\) corresponds to the second level of the Straubing-Thérien
hierarchy, an extensively studied hierarchy that is closely tied to the famous
dot-depth hierarchy.  A major open problem is to decide whether a given regular
language belongs to a given level of this hierarchy.  See the survey of
Pin~\cite{pin17} for a modern account on this topic and the recent major
progress of Place and Zeitoun~\cite{placez19}.

\paragraph{Organization of the paper.}  We introduce circuits, logic, and a bit
of algebra in \Cref{sec:preliminaries}.  In \Cref{sec:limits}, we introduce
so-called \emph{limits}, a classical tool in devising lower bounds against
depth-3 circuits.  In \Cref{sec:flower}, we present a simple lower bound against
a language in \(\Sigma_2\); this serves as both a warm-up for the main proof and to
identify the difficulties ahead.  In \Cref{sec:main}, we prove our main result,
that is, the Neutral Straubing Property for \(\Sigma_2\).  In \Cref{sec:consequences},
we derive some consequences of our main result, in particular the Straubing
Property for \(\Delta_2\).  We conclude in \Cref{sec:concl}.

\section{Preliminaries}\label{sec:preliminaries}

We assume familiarity with regular languages, logic, and circuits, although we
strive to keep this presentation self-contained.  We write \(\Reg\) for the class
of regular languages.

\paragraph{Word, languages, neutral letters.}  Following Lothaire~\cite{lothaire97}, a
word \(v = a_1a_2\ldots a_n\), with each \(a_i\) in an alphabet \(A\), is a \emph{subword}
of a word \(v\) if \(v\) can be written as \(v = v_0a_1v_1a_2\cdots a_nv_n\), with each
\(v_i\) in \(A^*\).  We say that a language \(L\) \emph{separates} \(X\) from \(Y\) if
\(X \subseteq L\) and \(L \cap Y = \emptyset\).  A language \(L\) has a \emph{neutral letter} if there is
a letter \(c\) such that \(u\cdot c\cdot v \in L \Leftrightarrow u\cdot v\in L\) for all words
\(u, v\).  We write \(\NL\) for the class of languages with a neutral letter.

\paragraph{Monoids, ordered monoids, morphisms.}  A \emph{monoid} is a set
equipped with a binary associative operation, denoted multiplicatively, with a
identity element.  An \emph{idempotent} of \(M\) is an element \(e \in M\) that
satisfies \(e^2=e\).
For an alphabet \(A\), the set \(A^*\) is the free monoid generated by \(A\), its
identity element being the empty word.  An \emph{ordered monoid} is a monoid
equipped with a partial order \(\leq\) compatible with the product, i.e.,
\(x \leq y\) implies \(xz \leq yz\) and \(zx \leq zy\) for any \(x,y,z\in M\).  Any monoid can be
seen as an ordered monoid, using equality as order.  An \emph{upper set} of an
ordered monoid \(M\) is a set \(S\) such that for any \(x,y \in M\), if
\(x \in S\) and \(x \leq y\) then \(y \in S\).  A morphism is a map
\(h\colon M \to N\) satisfying \(h(ab) = h(a)h(b)\) and \(h(1) = 1\), with
\(a, b \in M\) and \(1\) denoting the identity element of \(M\) and \(N\).

\paragraph{Monoids as recognizers.}  An ordered monoid \(M\) \emph{recognizes} a
language \(L \subseteq A^*\) if there is a morphism \(h\colon A^* \to M\) and an upper set
\(P\) of \(M\) such that \(L = h^{-1}(P)\).  The \emph{ordered syntactic monoid} of
\(L\) is the smallest ordered monoid that recognizes \(L\); it is finite iff
\(L\) is regular, in which case it is unique.\footnote{The ordered syntactic
  monoid is usually defined as the quotient of \(A^*\) by the so-called syntactic
  order induced by \(L\); the definition proposed here is
  equivalent~\cite[Corollary~4.4]{pin95} and allows to introduce one fewer
  concept.}

\paragraph{Logic.}  We work with first-order logics recognizing languages.  For
instance, the formula over the alphabet \(\{a, b\}\)
\[(\forall x)(\exists y)[\predmod_2(x) \lor (y = x + 1) \land (a(x) \leftrightarrow
b(y))]\]%
asserts that, in a given word \(w\), for every position \(x\) there is another
position \(y\) such that either \(x\) is divisible by 2 (\(\predmod_2\)) or \(y\) is just
after \(x\) and \(w\) has different letters at \(x\) and \(y\).  The predicates
\(\predmod_2\) and \(+1\) are examples of \emph{numerical predicates}, i.e., they only
speak about the numerical positions, not the contents of the input word.  The
predicates \(a(\cdot)\) and \(b(\cdot)\) are the \emph{letter predicates}.

In this paper, first-order logics are specified by restricting two aspects:
\begin{itemize}
\item The number of quantifier alternations.  We write \(\Sigma_2\) for the subset of
  first-order formulas that can be written as
  \((\exists x_1, x_2, \ldots)(\forall y_1, y_2, \ldots)[\phi]\) with \(\phi\) a quantifier-free formula, that
  is, \(\Sigma_2\) is the set of formulas starting with an existential quantifier and
  alternating \emph{once}.  In \Cref{sec:delta}, we will briefly mention
  \(\Pi_2\) (defined as \(\Sigma_2\) but with \(\exists\) and \(\forall\) swapped) and
  \(\Delta_2\), the set of formulas that are equivalent to \emph{both} a
  \(\Sigma_2\) and a \(\Pi_2\) formula (this describes, a priori, fewer languages than
  \(\Sigma_2\) or \(\Pi_2\)).
\item The \emph{numerical predicates} allowed.  Except for a quick detour in
  \Cref{sec:delta}, we will only be using two sets: \(<\), that is, the sole order
  relation (e.g., \(\sto\)) and \(\arb\) the set of \emph{all} predicates (e.g.,
  \(\sta\)).  In the latter set, there would be predicates asserting that two
  positions are coprime or that a position encodes an halting Turing machine,
  there is no restriction whatsoever.
\end{itemize}

The language of a formula is the set of words that satisfy it.  We commonly
identify a class of formulas with the class of languages they recognize.

\paragraph{Languages of a syntactic ordered monoid.} Let \(M\) be an ordered
monoid.  For any element \(x \in M\), the \emph{up-word problem} for \(M\) and
\(x\) is the following language over \(M\) seen as an alphabet:
\[\{w \in M^* \mid w \text{ evaluates in $M$ to an element} \geq x\}.\] In some precise
sense, a regular language has the same complexity as the hardest of the up-word
problems for its ordered syntactic monoid; in the case on \(\sta\), we can
state:
\begin{lemma}[{From \cite[Lemma 5.6]{pin95}}]\label{lem:regtomon}
  A regular language with a neutral letter is in \sta iff all the up-word
  problems for its ordered syntactic monoid are in \sta.
\end{lemma}

\paragraph{Circuits.}  We will study languages computed by families of
constant-depth, polynomial-size circuits consisting of unbounded fan-in
\(\land\)- and \(\lor\)-gates.  A circuit with \(n\) inputs
\(x_1, x_2, \ldots, x_n\) in some alphabet \(A\) can query whether any input contains any
given letter in \(A\).  The \emph{depth} of the circuit is the maximal number of
gates appearing on a path from an input to the output.  A circuit family is an
infinite set \((C_n)_{n\geq 0}\) where the circuit \(C_n\) has \(n\) inputs and one
output gate; a word \(w\) is deemed accepted if the circuit \(C_{|w|}\) outputs 1
when \(w\) is placed as input.  We identify classes of circuits with the class of
languages they recognize.

For a circuit \(C\) we visualize the inputs on \emph{top} and the one output gate
at the \emph{bottom}.\footnote{The literature has been flip-flopping between
  putting the inputs at the bottom or at the top, with a semblance of stability
  for ``top'' achieved in the late 90s.  This explains that some references
  mention ``bottom fan-in.''}  The \emph{top fan-in} of \(C\) is the maximum
fan-in of the gates that receive an input letter directly.  We say that \(C\) is a
\EAE circuit if it is layered with a bottom OR gate with AND gates as inputs,
each of these having OR gates as inputs.

We let \(\Sigma_2\) be the class of families of \EAE circuits of polynomial-size and
constant top fan-in.  In the literature, this circuit family is also called
\(\Sigma\Pi\Sigma(k)\) in~\cite{dingrg17} and \(\Sigma_2^{\text{poly},k}\)
in~\cite{caich98}, where \(k\) is the top fan-in.  It is a
subclass of \(\AC^0_3\), the class of polynomial-size, depth-3 circuits, a class
we will discuss in \Cref{sec:conscirc}.  The name ``\(\Sigma_2\) circuit family'' is
all the more justified that:
\begin{lemma}[{From~\cite[Proposition~11]{macielpt00}}]
  A language is recognized by a \(\Sigma_2[\arb]\) formula iff it is recognized by a
  \(\Sigma_2\) circuit family.
\end{lemma}
We will often exploit this equivalence without pointing at this lemma.

\paragraph{A decidable characterization of \(\Sigma_2[<]\).}  
Let \(x, y \in M\); we say that \(y\) is a \emph{subword} of \(x\) if there are two
words \(w_x, w_y \in M^*\) that evaluate, using \(M\)'s product, to \(x\) and
\(y\), respectively, and \(w_y\) is a subword of \(w_x\).  The following is a
characterization of the languages in \(\sto\) that is due to Pin and
Weil~\cite{pinw97}, and we use a version
proposed by Bojańczyk:
\begin{theorem}[From \cite{bojanczyk09}]\label{thm:eqs}
  A regular language is in \sto iff its ordered syntactic monoid \(M\) is such
  that for any \(x, y \in M\) such that \(x\) is an idempotent\footnote{This is
    usually written by letting \(x\) be any element, and considering
    \(x^\omega\), which is the unique idempotent that is a power of \(x\); we simplify
    the presentation slightly by simply requiring \(x\) to be an idempotent.}
  and \(y\) a subword of \(x\), it holds that
  \[x \leq x y x\enspace.\]
\end{theorem}

Equivalently, this could be worded over languages directly: a regular language
\(L\) is in \sto iff for any three words \(w_1, w_2, w_3\) that map to the same
idempotent in \(M\) and \(v\) a subword of \(w_2\), if a word
\(w_0\cdot w_1w_2w_3\cdot w_4\) is in \(L\), for some words \(w_0, w_4\), then so is
\(w_0\cdot w_1vw_3\cdot w_4\).  We will allude to this wording in some proofs in the
Consequences section (\Cref{sec:conscirc}).

\section{Limits and lower bounds against \sta}\label{sec:limits}

We present a tool that has been used several times to show lower bounds against
depth-3 circuits, in particular in \cite{hastadjp95}.
The following definition is attributed to Sipser therein:
\begin{definition}[From~\cite{sipser84}]
  Let \(F\) be a set of words, all of same length \(n\), and \(k > 0\). A
  \emph{\(k\)-limit} for \(F\) is a word \(u\) of length \(n\) such that for any set of
  \(k\) positions, a word in \(F\) matches \(u\) on all these positions.  In symbols,
  \(u\) satisfies:
  \[(\forall P \subseteq [n] . |P| = k)(\exists v \in F)(\forall p \in P)\left[u_p = v_p\right].\]
\end{definition}

Naturally, we will be interested in \(k\)-limits that fall outside of \(F\),
otherwise finding \(k\)-limits is trivial.  In fact, we will consider sets \(F\)
that are included in a subset of a target language, and find \(k\)-limits outside
of the target language itself.  We include a short proof of the following
statement for completeness and because it makes the statement itself more
readily understandable.
\begin{lemma}[{From~\cite[Lemma~2.2]{hastadjp95}}]\label{lem:klimitimpliesbad}
  Let \(L\) be a set of words all of same length \(n\) and \(C\) be a
  \EAE circuit that accepts at least all the words of \(L\).  Let \(k\) be the top
  fan-in of \(C\) and \(s\) its size.

  Assume there is a subset \(L' \subseteq L\) such that for any
  \(F \subseteq L'\) of size at least \(|L'|/s\) there is a \(k\)-limit for
  \(F\) that does not belong to \(L\).  Then \(C\) accepts a word outside of \(L\).  The
  hypothesis can be represented graphically as:
\end{lemma}
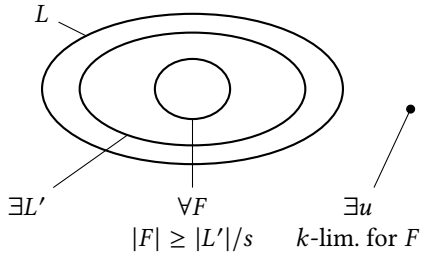
\begin{center}
  \begin{tikzpicture}
    \node (l') {\(\exists L'\)};
    \node[align=center,anchor=north] (f) at ($(l'.north)+(2.2cm,0)$) {\(\forall F\)\\\(|F| \geq |L'|/s\)};
    \node[align=center,anchor=north] (u) at ($(f.north)+(2.2cm,0)$) {\(\exists u\)\\\(k\)-lim.~for \(F\)};

    \path let \p1 = (f) in coordinate (center) at (\x1, \y1 + 1.8cm);
    
    \node (l) at ($(center) + (-2cm, 1cm)$) {\(L\)};

    \draw (l) -- (center);
    \draw[fill=white,thick] (center) ellipse (2cm and 1cm);

    \draw (l') -- (center);
    \draw[fill=white,thick] (center) ellipse (1.5cm and 0.75cm);

    \draw (f) -- (center);
    \draw[fill=white,thick] (center) ellipse (0.5cm and 0.4cm);

    \draw (u) -- +(0.7cm, 1.5cm) node [draw,fill=black,circle,minimum
    size=3pt,inner sep=0pt]
    {};
  \end{tikzpicture}
\end{center}
\begin{proof}
  At the bottom of \(C\), we have an OR gate of fan-in at most \(s\) that receives
  the result of some AND gates.  By counting, one of these AND gates should
  accept a subset \(F\) of \(L'\) of size at least \(|L'|/s\); we will now focus on
  that gate.  Let \(u \notin L\) be the \(k\)-limit for \(F\) that exists by hypothesis.
  Consider an OR gate that feeds into the AND gate under consideration.  This OR
  gate checks the contents of a subset \(P \subseteq [n]\) of \(k\) positions of the input.
  By hypothesis, there is a word \(v\) in \(F\) that matches \(u\) on all the
  positions in \(P\), hence the OR gate cannot distinguish between \(u\) and
  \(v\) and must output 1 (true) as \(v\) must be accepted.  This holds for all the
  OR gates feeding into the AND gate under consideration, hence the AND gate
  must accept \(u\), and so does \(C\).
\end{proof}

\begin{corollary}\label{cor:lb}
  Let \(L\) be a language and write \(L_n\) for the subset of words of length
  \(n\) in \(L\).  Assume that for any \(k, d \in \bbN\), there is an
  \(n \in \bbN\) and a subset \(L' \subseteq L_n\) such that every subset
  \(F \subseteq L'\) of size at least \(|L'|/n^d\) admits a \(k\)-limit outside of
  \(L\).  Then \(L\) is not in \(\sta\).
\end{corollary}
\begin{proof}
  For a contradiction, assume there is a \(\Sigma_2\) circuit family for \(L\), with
  top fan-in \(k\) and size \(n^d\).  Let \(n\) be the value provided by
  the hypothesis, then the circuit \(C\) for \(L_n\) satisfies the hypotheses of
  \Cref{lem:klimitimpliesbad}, hence \(C\) accepts a word outside of \(L\), a
  contradiction.
\end{proof}

\section{Warm-up: \(K = (ac^*b+c)^* \notin \sta\)}\label{sec:flower}

In this section, we follow the approach of Håstad, Jukna, and
Pudlák~\cite{hastadjp95} to show the claim of the section
title.  We present it in a specific way that will help us stress the commonality
and the differences of this approach with our main proof.

To show the claim of the section title, we consider a slightly different
language.  For any \(n\) that is a perfect square, we let \(\good_n\) be the set of
words of length \(n\) over \(\{a, b\}\) of the following shape:
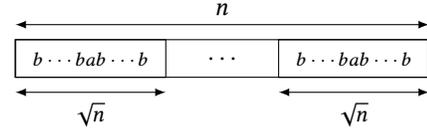
\begin{center}
  \begin{tikzpicture}
    \coordinate (ur) at (5.5cm, 0cm) {};
    \coordinate (ll) at (0cm, -0.5cm) {};
    \draw (ur) rectangle (ll);
    \draw[latex-latex,transform canvas={yshift=0.2cm}] (0,0) -- node [above]
    {\(n\)} (ur);

    \def\l{2cm}

    \foreach \t in {(0,0),{($(ur)-(\l,0)$)}} {
      \draw let \p1 = (ll) in
        \t rectangle +(\l, \y1)
         node[pos=.5,font=\scriptsize] (c) {\(b\cdots bab\cdots b\)};
      \draw[latex-latex] let \p1 = (c) in
        ($(\x1-\l/2,2*\y1-0.2cm)$) -- node[below,font=\small] {\(\sqrt{n}\)} +(\l,0);
    }
    
    \path let \p1 = (ll) in
      node at ($(ur)!.5!(0,0) + (0,\y1/2)$) {\(\ldots\)};
  \end{tikzpicture}
\end{center}
In words, a word is in \(\good_n\) if it can be decomposed into \(\sqrt{n}\)
\emph{blocks} of length \(\sqrt{n}\), such that each of them has exactly one \(a\).
We let \(\good = \bigcup_{n}\good_n\).

\begin{lemma}\label{lem:ktogood}
  If \(K\) is in \sta, then so is \good.
\end{lemma}
\begin{proof}
  This is easier to see on circuits, so assume there is a \(\Sigma_2\) circuit family
  for \(K\).  For \(n\) a perfect square, we design a circuit for \(\good_n\).  On any
  input, we convert the \(b\)'s to \(c\)'s and insert a \(b\) every \(\sqrt{n}\)
  positions; we call this the \emph{expansion} of the input word.  For instance,
  with \(n = 9\), the input \(abb\,bab\,bba\) is expanded to
  \(accb\,cacb\,ccab\).  Clearly, if the input word is in \(\good_n\), then its
  expansion is in \(K\).  Conversely, if a block of the input had two \(a\)'s, the
  expansion will not add a \(b\) in between, so the expansion is not in \(K\);
  similarly, if a block of the input contains only \(b\)'s, it will be expanded to
  only \(c\)'s sandwiched between two \(b\)'s, and the expansion will not be in
  \(K\) (in the case where the block containing only \(b\)'s is the first one, the
  expansion starts with \(c\cdots cb\), again putting the expansion outside of~\(K\)).

  Thus a circuit for \(\good_n\) can be constructed by computing the expansion
  (this only requires wires and no gates), then feeding that expansion to a
  circuit for \(K\).  If the circuit family for \(K\) were in \(\Sigma_2\), so would the
  circuit family for \(\good\).
\end{proof}

We use \Cref{cor:lb} to show that \(\good \notin \sta\).  Let then \(k, d \in \bbN\).  The
value of \(L'\) in \Cref{cor:lb} will simply be \(\good_n\), and we show:
\begin{lemma}\label{lem:bigtoklim}
  If \(n\) is large enough, any subset \(F \subseteq \good_n\) with
  \(|F| > k^{\sqrt{n}}\) has a \(k\)-limit outside of \(\good_n\).  This holds in
  particular if \(|F| \geq |\good_n|/n^d\).
\end{lemma}
\begin{proof}
  We rely on the \emph{Flower Lemma}, a combinatorial lemma that is a relaxation
  of the traditional Sunflower Lemma.  We first need to introduce some
  vocabulary.

  We consider families \(\cF\) containing sets of size \(s\) for some~\(s\).  The
  \emph{core} of the family is the set \(Y = \bigcap_{S \in \cF} S\).  The \emph{coreless}
  version of \(\cF\) is the family \(\cF_Y = \{S \setminus Y \mid S \in \cF\}\).  A set
  \(S\) \emph{intersects} a family \(\cF\) if all the sets of \(\cF\) have a nonempty
  intersection with \(S\).  Finally, a \emph{flower} with \(p\) petals is a family
  \(\cF\) of size \(p\) with core \(Y\) such that any set which intersects
  \(\cF_Y\) is of size at least \(p\).
  \begin{nest}
    \begin{lemma}[Flower Lemma~{\cite[Lemma 6.4]{jukna11}}]
      Let \(\cF\) be a family containing sets of cardinality \(s\) and
      \(p \geq 1\) be an integer.  If \(|\cF| > (p -1)^s\), then there is a subfamily
      \(\cF' \subseteq \cF\) that is a flower with \(p\) petals.
    \end{lemma}
  \end{nest}

  To apply this lemma, consider the mapping \(\tau\) from words in \(\good_n\) to
  \(2^{[n]}\) that lists all the positions where a word has an \(a\).  For instance,
  with \(n = 9\), \(\tau(bba\,abb\,bab) = \{3, 4, 8\}\).  For any word \(w\) in
  \(\good_n\), \(\tau(w)\) is of size \(\sqrt{n}\).  We let \(\cF = \{\tau(w) \mid w \in F\}\).
  
  We now apply the lemma with \(s = \sqrt{n}\) and \(p = k + 1\).  Since \(|\cF| =
  |F|\), we can apply the lemma on \(\cF\) and obtain a subfamily \(\cF'\) that is a
  flower with \(k+1\) petals.  Let \(Y\) be its core.  Consider the word \(u\) of
  length \(n\) over \(\{a, b\}\) which has \(a\)'s exactly at the positions in \(Y\).
  Then:
  \begin{itemize}
  \item \(u\) is outside of \(\good_n\).  Indeed, \(|Y| < \sqrt{n}\), since it is the
    intersection of distinct sets of size \(\sqrt{n}\).  Hence one of the blocks
    of \(u\) will contain only \(b\)'s, putting it outside of \(\good_n\).
  \item \(u\) is a \(k\)-limit.  Let \(P\) be a set of \(k\) positions, we will find a
    word that is mapped to \(\cF'\) that matches \(u\) on \(P\).  If a position in \(P\)
    points to an \(a\) in \(u\), then every word in \(\cF'\) has an \(a\) at that
    position (by construction, since this position would belong to the core
    \(Y\)).  So we assume that \(P\) contains only positions on which \(u\) is \(b\).
    Since \(|P|\) is \(k\), it cannot intersect \(\cF'\), hence there is a set \(S \in
    \cF'\) such that \(S \cap P = \emptyset\).  The set \(S\) is thus \(\tau(w)\) for a word \(w \in F\)
    that has a \(b\) on all positions in \(P\).  This word \(w\) thus matches \(u\) on
    \(P\), concluding the proof of the main statement.
  \end{itemize}

  The ``in particular'' part is implied by the fact that, for \(n\) large enough:
  \[\frac{|\good_n|}{n^d} = \frac{\sqrt{n}^{\sqrt{n}}}{n^d} \geq
  k^{\sqrt{n}}.\]
\end{proof}

\begin{theorem}\label{thm:k}
  The language  \(K = (ac^*b+c)^*\) is not in \sta.
\end{theorem}
\begin{proof}
  \Cref{cor:lb} applied on \(\good\), using \Cref{lem:bigtoklim}, implies that
  \(\good \notin \sta\).  \Cref{lem:ktogood} then asserts that \(K\) cannot be in
  \(\sta\) either.
\end{proof}

\section{The regular languages with a neutral letter not in \sto are not in
  \sta}\label{sec:main}

The proof of the statement of the section title is along two main steps:
\paragraph{\Cref{sec:nostoandstatosep}.} We will start with a language with a
neutral letter \(L \notin \sto\).  Since it is not in \sto, there are
\(x, y \in M\) that falsify the equations of \Cref{thm:eqs}.  We use these witnesses
to build a up-word problem \(T\) of the ordered syntactic monoid of \(L\) and show
that it lies outside of \(\sta\), implying that \(L \notin \sta\) by \Cref{lem:regtomon}.

To show \(T\) out of \sta, we identify (\Cref{sec:somewords}) a subset of
well-behaved words of \(T\), and make some simple syntactical changes (in
\Cref{sec:tgoodtogood}) on them so that they look like words in \good, in a
similar fashion as the ``expansions'' of \Cref{lem:ktogood}.  The argument used
in \Cref{lem:ktogood} then needs to be refined, as we do not have that any word
outside of \good comes from a word outside of \(T\).  We will define a set \bad of
words that look like words in \good except for \emph{one} block that contains
only \(b\)'s; \Cref{lem:ktogood} is then worded as: if \(K\) is in \sta, then there
is a \sta language that separates \good from \bad (\Cref{lem:ttogoodbad}).

\paragraph{\Cref{sec:nogoodbadsepinsta}.} We show that no language of \sta can
separate \good from \bad.  We thus need to provide a statement in the spirit of
\Cref{lem:bigtoklim}.  We first write good and bad words in a succinct
(``packed'') way, as words in \([\sqrt{n}]^{\sqrt{n}}\), the \mbox{\(i\)-th} letter
being some value \(v\) if the original word had the \(a\) of its \(i\)-th block in
position \(v\) (\Cref{sec:packed}).  We then translate the notion of \(k\)-limit to
packed words (\Cref{lem:packedlimit}).  Finally, we provide a measure of how
diverse a set of (packed) good words is (\Cref{def:entail}), and show that such
a set is either not diverse and small (\Cref{lem:tangledimpliessmall}), or
diverse and admits a \(k\)-limit (\Cref{lem:nottangledimpliesklim}).  Our term for
``not diverse'' will be \emph{tangled}, referring to the fact that there is a
strong correlation between the contents of positions within words.

\subsection{If a language not in \sto is in \sta, we can separate \good from
  \bad with a language in \sta}\label{sec:nostoandstatosep}

\subsubsection{Target up-word problem and some of its words}\label{sec:somewords}

For the rest of this section, let \(L \subseteq A^*\) be a regular language with a neutral
letter that lies outside of \sto and let \(M\) be its ordered syntactic monoid.
Since \(L\) is not in \sto, there are elements \(x, y \in M\) such that
\(x \not\leq x y x\) with \(x\) an idempotent and \(y\) a subword of \(x\).  Let
\(T\) be the up-word problem of \(M\) for \(x\).  Clearly, any word of
\(M^*\) that evaluates to \(x y x\) does \emph{not} belong to \(T\).

Naturally, \(y\) is thus also a subword of \(x\); this provides us with words
that evaluate to \(x\) and \(y\) of the shape:
\[x_1 y_1 \dots x_t y_t \text{ evaluates to } x, \quad y_1 \dots y_t \text{
  evaluates to } y,\] with each
\(x_i\) and \(y_i\) in \(M\).  (Note that we can use the identity element of
\(M\) as needed to ensure we have as many \(x_i\)'s as \(y_i\)'s.)

Let \(n \in \bbN\) be a large enough perfect square (``large enough'' only depends
on \(t\), the length of the word for \(y\)).  We define \(\sqrt{n}+1\) words of length
\(\sqrt{n}+t\) over \(M\):
\begin{itemize}
\item For \(i \in [\sqrt{n}]\),
  \[x^{(i)} = \left(1^{i-1}x_11^{\sqrt{n}-i} \cdot y_1\right) \cdots
  \left(1^{i-1}x_t1^{\sqrt{n}-i}\cdot y_t\right).\]%
  Here \(1\in M\) is the neutral element of \(M\).  Note that these words evaluate to
  \(x\).
\item Additionally, we consider the word
  \(1^{\sqrt{n}}y_1\dots 1^{\sqrt{n}}y_t\), which evaluates to \(y\), and we simply
  write \(y\) for it.  Note that \(y\) can be obtained by removing all the letters
  \(x_j\) from any word \(x^{(i)}\).
\end{itemize}

Call \emph{\(T\)-good} a concatenation of \(\sqrt{n}\) words of the form
\(x^{(i)}\) sandwiched between two words \(x^{(1)}\) (the 1 is arbitrary), and
\emph{\(T\)-bad} a word obtained by changing, in a \(T\)-good word, \emph{exactly
  one} of the words \(x^{(i)}\) to \(y\) (but for the \(x^{(1)}\) at the beginning and
end).  By construction, any \(T\)-good word evaluates to \(x\) in \(M\), so belongs
to \(T\), while any \(T\)-bad word evaluates to \(x yx\), hence does not belong to
\(T\).  Note that if we had switched \emph{two} blocks of a \(T\)-good word to
\(y\), we would not be able to say whether it belonged to \(T\) or not.

\subsubsection{\(T\)-good, \(T\)-bad to \good and \bad}\label{sec:tgoodtogood}

The \(T\)-good and \(T\)-bad words contain a lot of redundant information, for
instance \(x^{(i)}\) is of length \(\sqrt{n}+t\), while all the information it
really contains is \(i \in [\sqrt{n}]\).  Recall the set \(\good_n\) of
\Cref{sec:flower} which contains all words of length \(n\) over \(\{a, b\}\) that
can be divided into \(\sqrt{n}\) blocks of length \(\sqrt{n}\), each containing a
single \(a\).  Again, we let \(\good\) be all such words, of any perfect square
length.  Define similarly \(\bad_n\) as the set of words that are like
\(\good_n\) except for \emph{one} block which has only \(b\)'s, and let
\(\bad = \bigcup_n \bad_n\).

In the next lemma, we show that we can modify, using only wires in a circuit,
words over \(\{a,b\}\) so that if they are in \good they become \(T\)-good, and if
they are in \bad they become \(T\)-bad.  This modification is simple enough that
we can take a \(\sta\) circuit family for \(T\), apply the modification at the top
of each circuit, and still have a circuit family in \(\sta\); the resulting
circuit family separates \good from \bad:
\begin{lemma}\label{lem:ttogoodbad}
  If \(\,T \in \sta\), then there is a \(\sta\) language that separates \(\good\) from
  \(\bad\).
\end{lemma}
\begin{proof}
  As in \Cref{lem:ktogood}, this is easier seen on circuits: we design a circuit
  for inputs of length \(n\) over \(\{a, b\}\) that separates \(\good\) from \(\bad\).

  Consider the first block of \(\sqrt{n}\) letters of the input.  We replicate it
  \(t\) times, with the \(i\)-th replication changing \(b\)'s to \(1\) and
  \(a\)'s to \(x_i\).  We then concatenate these and add \(y_i\) between the
  \(i\)-th and \((i+1)\)-th replication.  For instance, \(b^7ab^{\sqrt{n}-8}\) would
  turn into:
  \[\left(1^7x_11^{\sqrt{n}-8}\cdot y_1\right) \cdots \left(1^7x_t1^{\sqrt{n}-8}\cdot
    y_t\right) = x^{(8)}.\]%
  In particular, if the block were all \(b\)'s, we would obtain the word \(y\),
  which has no letter \(x_j\).

  We can do this to each block of \(\sqrt{n}\) letters, concatenate the resulting
  words, then add the word \(x^{(1)}\) at the beginning and the end.  Note that
  these operations can be done with only wires, with no gates involved.

  If the input word is in \(\good\), then the word produced is \(T\)-good, hence in
  \(T\).  If it was in \(\bad\), then the resulting word would be \(T\)-bad, hence
  would lie outside of \(T\).  This shows that the desired circuit can be
  constructed using the above wiring followed by the circuit for \(T\) for inputs
  of length \((t\sqrt{n} + t)(2 + \sqrt{n})\).  Since \(t\) is a constant and
  depends solely on \(L\), the resulting circuit is of polynomial size and of the
  correct shape.
\end{proof}

\subsection{No language in \sta separates \good from \bad}\label{sec:nogoodbadsepinsta}

Note that this section is independent from the previous one.  We will now rely
on \Cref{cor:lb} to show that any \sta language \(L\) that accepts all of \good
must accept a word in \bad.  To apply \Cref{cor:lb}, \emph{from this point
  onward} we let \(k, d \in \bbN\), and set \(n\) to be a large enough value that
depends only on \(k\) and \(d\).  The role of \(L'\) in the statement of \Cref{cor:lb}
will be played by \(\good_n\) and we will build \(k\)-limits belonging to
\(\bad_n\), which we call \emph{bad \(k\)-limits}.  The reader may check that the
statements of the forthcoming \Cref{lem:tangledimpliessmall} and
\Cref{lem:nottangledimpliesklim} conclude the proof.

\subsubsection{Packed words}\label{sec:packed}

Words in \(\good_n\) and \(\bad_n\) can be described by the position of the letter
\(a\) in each block of size \(\sqrt{n}\).  We make this explicit, by seeing
\([\sqrt{n}, \bot] = [\sqrt{n}] \cup \{\bot\}\) as an alphabet, and working with words in
\([\sqrt{n}, \bot]^{\sqrt{n}}\).  We call these words \emph{packed} and will use
Greek letters \(\lambda, \mu, \nu\) for them; we also call the letter at some position in
packed words its \emph{contents} at this position, only to stress that we are
working with packed words. We define the natural functions to pack and unpack
words:
\begin{itemize}
\item \(\unpack\colon [\sqrt{n}, \bot] \to \{a, b\}^{\sqrt{n}}\) maps \(i\) to
  \(b^{i-1}ab^{\sqrt{n}-i}\) and \(\bot\) to \(b^{\sqrt{n}}\).  This extends naturally to
  \emph{words} over \([\sqrt{n}, \bot]\).
\item \(\pack\colon \{a, b\}^* \to [\sqrt{n}, \bot]^*\) is the inverse of
  \(\unpack\).  We will use that function on sets of words too, with the natural
  meaning.
\end{itemize}

\begin{example}
  With \(n = 9\), \(\pack(abb\,bbb\,bab) = 1\bot2\), and \(\unpack(31\bot)=bba\,abb\,bbb\).
\end{example}

We can now rephrase the notion of \(k\)-limit using packed words:
\begin{lemma}\label{lem:packedlimit}
  Let \(F \subseteq \good_n\) and define \(\Phi = \pack(F)\).  If \(\mu\) is a packed word that has
  the following properties, then \(\unpack(\mu)\) is a bad \(k\)-limit for
  \(F\):
  \begin{enumerate}
  \item There is a word \(\nu \in \Phi\) that differs on a single position \(i\) with
    \(\mu\), at which \(\mu\) has contents \(\bot\):
    \[\mu_i = \bot \land (\forall j \neq i)[\nu_j = \mu_j].\]
  \item For every set \(C \subseteq [\sqrt{n}]\) of contents that contains
    \(\nu_i\) and every set \(P \subseteq [\sqrt{n}] \setminus \{i\}\) of positions such that
    \(|C| + |P| = k\), there is a word \(\lambda \in \Phi\) whose contents at position
    \(i\) is not in \(C\) and that matches \(\nu\) on \(P\):
    \[\lambda_i \notin C \land (\forall p \in P)[\lambda_p = \nu_p].\]
  \end{enumerate}
\end{lemma}
\begin{proof}
  Write \(u\) for \(\unpack(\mu)\).  That \(u \in \bad\) is immediate from Property~1:
  \(u\) is but a word \(v\) of \(\good\) in which one block was set to all \(b\)'s.

  We now show that \(u\) is a \(k\)-limit.  Let \(T\) be a set of \(k\) positions, we
  split \(T\) into two sets:
  \begin{itemize}
  \item \(T'\) is the set of positions that do not belong to the \(i\)-th block of
    \(u\), that is, they do not satisfy \(\lceil p / \sqrt{n}\rceil = i\).  We let \(P\) be each
    of the elements of \(T'\) divided by \(\sqrt{n}\), that is, for any \(p \in T'\) we
    add \(\lceil p / \sqrt{n}\rceil\) to \(P\).
  \item \(T''\) is the set of positions that \emph{do} fall in the \(i\)-th block.
    Note that \(u\) only has \(b\)'s at the positions of \(T''\).  We let \(C\) be that
    set, modulo \(\sqrt{n}\), that is, for any \(p \in T''\), we add
    \(p \mathbin{\text{mod}} \sqrt{n}\) to \(C\) or \(\sqrt{n}\) if this value is \(0\).
  \end{itemize}

  First, if \(\mu_i \notin C\), then \(T\) indicates positions of \(u\) that have the same
  letter as in \(\unpack(\nu) \in F\), so a word of \(F\) matches \(u\) over
  \(T\), as required.  We thus assume next that \(\mu_i \in C\).

  Let \(\lambda \in \Phi\) be the word given by Property~2 for \(C\) and \(P\), we claim that \(w
  = \unpack(\lambda)\) matches \(u\) on the positions of \(T\), concluding the proof.

  First note that the \(i\)-th block of \(w\) has its \(a\) in a position that is not
  in \(T''\), hence \(w\) matches \(u\) on \(T''\).  Consider next any position
  \(p \in T'\) and write \(j\) for the block in which \(p\) falls (i.e.,
  \(j = \lceil p/\sqrt{n}\rceil\)).  Since \(\mu_j = \lambda_j\) by hypothesis, the
  \(j\)-th block of \(u\) and \(w\) are the same, hence \(u_p = w_p\).
\end{proof}

\subsubsection{Tangled sets of good words are small, nontangled ones have a bad
  \(k\)-limit}

Consider a \(F \subseteq \good_n\).  To find a bad \(k\)-limit for \(F\), we need a lot of
diversity in \(F\); see in particular Prop.~2 of \Cref{lem:packedlimit}.
Hence having some given contents at a given position in a word of \(\Phi\) should not
force too many other positions to have a specific value.  We make this notion
formal:

\begin{definition}\label{def:entail}
  Let \(F \subseteq \good_n\) and \(\Phi = \pack(F)\).  The \emph{entailment relation} of
  \(\Phi\) is relating sets of pairs \((i, c)\) of position/contents in words of
  \(\Phi\).  Let us say that a word and a pair position/contents \((i, c)\)
  \emph{agree} if the contents at position \(i\) of the word is \(c\), and that a
  word and a set of such pairs agree if they agree on \emph{all} the pairs.  We say
  that a set of pairs is an \(i\)-set if all its pairs have \(i\) as position.

  A set \(S\) of pairs position/contents \emph{entails} an \(i\)-set \(D\)
  if all words in \(\Phi\) that agree with \(S\) also agree with \emph{at least one} pair of
  \(D\); additionally, the position \(i\) should not appear in \(S\):
  \[
  \begin{array}{l}
    (\not\exists c)[(i, c) \in S] \land \hfill\\
    (\forall \mu \in \Phi)[(\forall (j, d) \in S)[\mu_j = d] \rightarrow \mbox{}\\
    \qquad\qquad (\exists (i, c) \in D)[\mu_i = c]].
  \end{array}\]
  
  The set \(F\) is said to be \emph{\(k\)-tangled} if for any word \(\mu \in \Phi\) and any
  position \(i \in [\sqrt{n}]\), there is an \(i\)-set of pairs of size \(\leq k\) that contains
  \((i, \mu_i)\) and that is entailed by a set of size \(k\) that agrees with
  \(\mu\).  In other words, every position of \(\mu\) is entailed by a subset of its
  positions. We drop the \(k\) in \(k\)-tangled if it is clear from context.
\end{definition}

\begin{lemma}\label{lem:tangledimpliessmall}
  Let \(F \subseteq \good_n\).  If \(F\) is \(k\)-tangled, then
  \(|F| < \sqrt{n}^{2k\sqrt{n}/(2k+1)}\).  In particular, \(|F| < |\good_n| / n^d\).
\end{lemma}
\begin{proof}

  Assume \(F\) is \(k\)-tangled and let \(\Phi = \pack(F)\).  We show that every word in
  \(\Phi\) can be fully described in \(\Phi\) by fully specifying a portion
  \(k/(k+1)\) of its positions and encoding the contents of each of the other
  \(1/(k+1)\) positions with elements from \([k]\). That is, if two words in \(\Phi\) have the same
  such description, they are the same, hence \(\Phi\) cannot be larger than the
  number of such descriptions.  We first show this property, then derive the
  numerical implication on \(|F|\).

  Let \(\mu \in \Phi\), we construct iteratively a set \(K\) of positions that we will
  fully specify and a set \(K^+\) of positions that are restricted when setting
  the positions in \(K\).

  First consider the pair \((1, \mu_1)\).  Since \(F\) is tangled, there is an 1-set
  containing \((1, \mu_1)\), entailed by a set \(S\) that agrees with \(\mu\).  We add to
  \(K\) the positions of \(S\) and to \(K^+\) the positions of \(S\) \emph{and} position
  \(1\).

  We now iterate this process: Take a pair \((i, \mu_i)\) such that
  \(i \notin K^+\).  There is an \(i\)-set containing \((i, \mu_i)\) that is entailed by a
  set \(S\) that agrees with \(\mu\).  Let \(S'\) be the set of positions of
  \(S\) that are not in \(K^+\).  We add \(S'\) to \(K\), and \(S' \cup \{i\}\) to
  \(K^+\).  Note that the size increase for \(K^+\) is one more than that for
  \(K\).  We continue iterating until all positions appear in \(K^+\).

  We now bound the size of \(K\) at the end of the computation.  For each
  iteration, in the worst case, we need to add \(k\) positions to \(K\) to obtain
  \(k+1\) new positions in \(K^+\) (this is the worst case in the sense that this is
  the worst ratio of the number of positions we need to pick in \(K\) to the
  number of positions that are put in \(K^+\)).  In that case, after \(s\) steps, we
  have \(|K| = sk\) and \(|K^+|=sk+s\).  Thus when \(|K^+| = \sqrt{n}\), that is, when
  no more iterations are possible, we have:
  \[sk+s = \sqrt{n} \Rightarrow s = \frac{\sqrt{n}}{k+1}.\]
  This shows that \(|K| \leq k\sqrt{n}/(k+1)\).

  We now turn to describing the word \(\mu\) using \(K\).  We first provide all the
  contents of \(\mu\) at positions in \(K\); call \(Z\) the set of pairs
  position/contents of \(\mu\) that correspond to positions in \(K\).  We mark the
  positions of \(K\) as \emph{specified}, and carry on to specify the other
  positions in a deterministic fashion.

  We first fix an arbitrary order on sets of pairs of position/contents.  We
  iterate through all the subsets of \(Z\) of size \(k\), in order.  For each such
  subset \(S\), we consider, in order again, the subsets \(D\) that are entailed by
  \(S\).  Assume \(D\) is an \(i\)-set; if position \(i\) is already specified, we do
  nothing, otherwise, we describe which element of \(D\) is \((i, \mu_i)\) using
  an integer in \([k]\), and mark \(i\) as specified.  We proceed until all the
  subsets of \(Z\) have been seen, at which point, by construction of \(K\), all the
  positions will have been specified.  As claimed, given \(Z\) and the description
  of which elements in sets \(D\) correspond to the correct contents, we can
  reconstruct \(\mu\).

  Summing up, to fully describe \(\mu\), we had to specify the positions of
  \(K\) (one of \(\binom{\sqrt{n}}{k\sqrt{n}/(k+1)}\) possible choices), their
  contents (one of \(\sqrt{n}^{k\sqrt{n}/(k+1)}\) possible choices), and
  for each position not specified by \(K\), we needed to provide an integer in
  \([k]\) (one of \(k^{\sqrt{n}/(k+1)}\) possible choices).  This shows that:
  \begin{align*}
    |F| & \le \binom{\sqrt{n}}{k\sqrt{n}/(k+1)} \cdot \sqrt{n}^{k\sqrt{n}/(k+1)} \cdot
    k^{\sqrt{n}/(k+1)}\\
    & \leq 2^{\sqrt{n}} \cdot 2^{(\sqrt{n}/(k+1)) (k \log{\sqrt{n}})} \cdot 2^{(\sqrt{n}/(k+1)) \log(k)}\\
    & = 2^{(\sqrt{n}/(k+1))( (k+1) + k\log{\sqrt{n}} + \log{k})}\\
    & \leq 2^{(\sqrt{n}/(k+1)) ((k+\frac{1}{k})\log{\sqrt{n}})} \tag{n\text{ large enough}}\\
    & \leq \sqrt{n}^{(k+\frac{1}{k})\sqrt{n}/(k+1)} \leq \sqrt{n}^{2k\sqrt{n}/(2k+1)}.
  \end{align*}

  The ``in particular'' part is a consequence of the fact that, for \(n\) large
  enough:
  \[\frac{|\good_n|}{n^d} = \frac{\sqrt{n}^{\sqrt{n}}}{n^d} \geq
  \sqrt{n}^{2k\sqrt{n}/(2k+1)}.\]
\end{proof}

\begin{lemma}\label{lem:nottangledimpliesklim}
  Let \(F \subseteq \good_n\). If \(F\) is not \(k\)-tangled, then \(F\) has a bad \(k\)-limit.
\end{lemma}
\begin{proof}
  Write \(\Phi = \pack(F)\).  That \(F\) is not tangled
  means that there is a word \(\nu \in \Phi\) and a position \(i\) such that for any set of
  pairs position/contents \(S\) and any \(i\)-set \(D\) that contains
  \((i, \nu_i)\), \(S\) does not entail \(D\).  We define \(\mu\) to be the word
  \(\nu\) but with \(\mu_i\) set to \(\bot\).  We show that \(\unpack(\mu)\) is a bad
  \(k\)-limit using \Cref{lem:packedlimit}.  Property~1 therein is true by
  construction, so we need only show Property~2.

  Let \(C \subseteq [\sqrt{n}]\) with \(\mu_i \in C\) and
  \(P \subseteq[\sqrt{n}] \setminus \{i\}\) with \(|C|+|P| = k\).  We add some more arbitrary
  positions in \(P\) so that \(|P| = k\), avoiding~\(i\).  Define:
  \[S = \{(p, \mu_p) \mid p \in P\}, \, D = \{(i, c) \mid c \in C\}.\] By hypothesis, since
  \((i, \mu_i) \in D\), \(S\) does not entail \(D\).  This means that there is a word
  \(\lambda \in \Phi\) such that \(S\) and \(\lambda\) agree, but
  \(\lambda_i \notin C\).  This is the word needed for Property~2 of \Cref{lem:packedlimit},
  concluding the proof.
\end{proof}

\begin{corollary}\label{cor:good}
  No \sta language can separate \(\good\) from \(\bad\).
\end{corollary}
\begin{proof}
  We apply \Cref{cor:lb} on any language \(L\) that separates \(\good\) from
  \(\bad\).  We let \(k, d \in \bbN\), and \(n\) large enough; \(L'\) in the statement of
  \Cref{cor:lb} is set to \(\good_n\).  We are then given a set \(F\) of size at
  least \(|\good_n|/n^d\) and \Cref{lem:tangledimpliessmall} shows that \(F\) is not
  tangled.  \Cref{lem:nottangledimpliesklim} then implies that \(F\) has a bad
  \(k\)-limit.  \Cref{cor:lb} concludes that \(G\) is not in \sta, showing the
  statement.
\end{proof}

\begin{theorem}[Neutral Straubing Property for \(\Sigma_2\)]\label{thm:main}
  \[\sta \cap \Reg \cap \NL \subseteq \sto.\]
\end{theorem}
\begin{proof}
  Let \(L \notin \sto\) with a neutral letter and \(T\) be the language defined in
  \Cref{sec:somewords}.  \Cref{cor:good} and \Cref{lem:ttogoodbad} imply that
  \(T\) cannot be in \sta, and in turn, \Cref{lem:regtomon} shows that \(L\) cannot
  be in \(\sta\).
\end{proof}

\section{Consequences}\label{sec:consequences}

\subsection{Life without neutral letters}\label{sec:delta}

The \emph{regular numerical predicates}, denoted \reg, are the numerical
predicates \(+1, <,\) and for any \(p > 0\), \(\predmod_p\) which is true of a
position if it is divisible by \(p\).  The term ``regular'' stems from the fact
that these are the properties on numerical positions that automata can express.
Recall that the Straubing Property for a logic \(\cL\) expresses that
\(\cL[\arb] \cap \Reg = \cL[\reg]\).

The Straubing Property does not immediately imply the Neutral Straubing
Property; for this, one would need in addition that
\(\cL[\reg] \cap \NL \subseteq \cL[<]\).  This latter property is called the \emph{Crane
  Beach Property} of \(\cL[\reg]\), and also stems from the natural idea that if a
language has a neutral letter, then numerical predicates do not provide any
useful information. Albeit natural, this property is false for
\(\FO[\arb]\)~\cite{barringtonilst05}, but \Cref{thm:main} shows that \(\Sigma_2[\reg]\) does have the
Crane Beach Property.

Relying on \Cref{thm:k} and some results from \cite{cadilhacp21}, we can show:
\begin{theorem}[Straubing Property of \(\Delta_2\)]\label{thm:strdelta2}
  \[\Delta_2[\arb]\cap \Reg = \Delta_2[\reg].\]
\end{theorem}
\begin{proof}
  The proof structure is as follows: We first show that
  \(\Delta_2[\arb] \cap \Reg\) has some closure properties, so that it is a so-called
  \emph{lm-variety}.  We then show that \(\Delta_2[\reg]\) recognizes precisely all the
  regular languages definable with a first-order formula with \emph{two}
  variables that uses \reg as numerical predicates; that class of languages is
  denoted \(\FO^2[\reg]\).  We then rely on the following lemma, where \(K
  = (ac^*b+c)^*\):
  \begin{nest}
    \begin{lemma}[{From \cite[Lemma~8]{cadilhacp21}}]\label{lem:fo2k}
      If an lm-variety of regular languages \(\cV\) satisfies:
      \[\FO^2[\reg] \subseteq \cV \subseteq \FO[\arb] \text{ and } K \notin \cV\]
      then \(\cV = \FO^2[\reg]\).
    \end{lemma}
  \end{nest}
  Since \(\Delta_2[\arb] \cap \Reg\) satisfies the hypotheses, it is equal to \(\FO^2[\reg]
  = \Delta_2[\reg],\) concluding the proof.

  \paragraph{\(\bm{\Delta_2[\mathbf{arb}] \cap \mathbf{Reg}}\) is an lm-variety.}  We ought to first define
  \emph{lm-variety}.  If for a morphism \(h\colon A^* \to B^*\) there is a
  \(k\) such that \(h(A) \subseteq B^k\), we call \(h\) an \emph{lm-morphism}, where lm stands
  for \emph{length-multiplying}.  Given a language \(L\) and a letter \(a\), the
  \emph{left quotient} of \(L\) by \(a\) is the set
  \(a^{-1}L = \{v \mid av \in L\}\).  The \emph{right quotient} \(La^{-1}\) is defined
  symmetrically.  An \emph{lm-variety of languages} is a set of languages closed
  under the Boolean operations, quotient, and inverse lm-morphisms.

  Since \(\Reg\) is an lm-variety of languages, it is sufficient to show that
  \(\Delta_2[\arb]\) is too; this is not hard:
  \begin{itemize}
  \item Boolean operations: Both \(\Sigma_2[\arb]\) and \(\Pi_2[\arb]\) formulas are closed
    under Boolean OR and AND, hence the classes of languages they recognize are
    closed under union and intersection, and so is \(\Delta_2[\arb]\).  Also,
    since the negation of a \(\Sigma_2[\arb]\) formula is a \(\Pi_2[\arb]\) formula, and
    vice versa, \(\Delta_2[\arb]\) is closed under complement.
  \item Quotient: We show that \(\Sigma_2[\arb]\) is closed under quotient; the proof
    is the same for \(\Pi_2[\arb]\), and this implies that \(\Delta_2[\arb]\) is also
    closed under quotient.  Let \(L \in \Sigma_2[\arb]\) and \(a\) be a letter.  Consider
    the circuit for the words of length \(n\) in \(L\).  We can hardwire the
    first letter to \(a\); the resulting circuit has \(n-1\) inputs, and recognizes
    a word \(w\) iff \(aw \in L\).  The family thus obtained recognizes \(a^{-1}L\).  The
    argument for right quotient is similar.
  \item Lm-morphisms: Again, we show this holds for \(\Sigma_2[\arb]\), the proof for
    \(\Pi_2[\arb]\) being similar, and these two facts imply closure under
    lm-morphisms of \(\Delta_2[\arb]\).  Let \(L \in \Sigma_2[\arb]\) over the alphabet
    \(B\) and \(h\) be an lm-morphism such that \(h(A) \subseteq B^k\) for some
    \(k\). Consider the circuit for the words of \(L\) of length \(kn\) for some
    \(n\).  Given a word in \(A^n\), we can use wires\mcnote{In fact, we need
      slightly more than wires, we may need small OR gates.  This is not a
      problem for \(\Sigma_2\), but is it for \(\Pi_2\)?} to map each input letter
    \(a \in A\) to \(h(a)\), and we can feed the resulting word to the circuit for
    \(L\).  A word \(w \in A^n\) is thus accepted iff \(h(w) \in L\), hence the circuit
    family thus defined recognizes \(h^{-1}(L)\).
  \end{itemize}

  \paragraph{\(\bm{\Delta_2[\mathbf{reg}]}\) and
    \(\bm{\mathbf{FO}^2[\mathbf{reg}]}\) recognize the same languages} We show the
  inclusion from left to right, the converse being similar.  We rely on the fact
  that \(\FO^2[<, +1] = \Delta_2[<, +1]\), a result due to Thérien and
  Wilke~\cite[Theorem~7]{therienw98}.  The rest of our proof is
  fairly simple: we put the information given by the \(\predmod_p\) predicates
  within the alphabet, show that this information is easily checked with one
  universal quantifier if we have \reg predicates, and that if the modular
  predicates are put within the alphabet, the only required predicates to
  express our \(\Delta_2[\reg]\) formula are \(<\) and \(+1\).  We then rely on the
  equivalence of \(\FO^2\) and \(\Delta^2\) over these predicates to conclude.

  Formally, let \(\phi \in \Delta_2[\reg]\) be a formula over the alphabet \(A\).  Let
  \(P\) be the set of moduli used in \(\phi\), that is, the predicate \(\predmod_p\) appears
  in \(\phi\) iff \(p \in P\).  The \emph{\(P\)-annotation} of a word
  \(w \in A^*\) is the word in \((A \cup 2^P)^*\) that indicates, for each position
  \(i\), the set of moduli in \(P\) that divide \(i\).  In other words, the
  \(P\)-annotation of \(w = w_1w_2\cdots w_n\), with each \(w_i \in A\), is the word of
  length \(n\) whose \(i\)-th letter is:
  \[\binom{w_i}{\{p \in P \mid p \text{ divides } i\}} \in (A \cup 2^P).\]
  Let \(\cW\) be the set of words in \((A \cup 2^P)^*\) that are \(P\)-annotations.
  Then:
  \begin{itemize}
  \item \(\cW \in \FO^2[\reg] \cap \Delta_2[\reg]\).  Indeed, a formula for \(\cW\) need only
    assert that for all positions \(i\), the \(2^P\) part of the letter at position
    \(i\) is exactly the set \(\{p \in P \mid p \text{ divides } i\}\).  This can be
    written as a single universal quantifier followed by a quantifier-free
    formula.
  \item There is a \(\phi' \in \Delta_2[<, +1]\) such that the words of \(\cW\) that satisfy
    \(\phi'\) are precisely the \(P\)-annotations of words that satisfy
    \(\phi\).  The formula \(\phi'\) is simply the formula \(\phi\) in which each predicate
    \(\predmod_p(x)\) is replaced with the property ``\(p\) belongs to the \(2^P\) part of
    the letter at position \(x\),'' which can be written as a simple disjunction.
    If the input word is a \(P\)-annotation, \(\predmod_p(x)\) is indeed equivalent to
    that property.
  \end{itemize}
  Since \(\FO^2[<, +1] = \Delta_2[<, +1]\), there is a formula \(\psi\) of \mbox{\(\FO_2[<,+1]\)} that
  accepts the same language as \(\phi'\).  With \(\psi_\cW\) the \(\FO^2[\reg]\) formula for
  \(\cW\), we conclude that \(\psi \land \psi_\cW\) is a \(\FO^2[\reg]\) formula that is
  equivalent to \(\phi\).
\end{proof}

The proof of the previous statement hinged on the characterization given by
\Cref{lem:fo2k}.  To show the (nonneutral) Straubing Property of
\(\Sigma_2[\arb]\), a similar statement will need to be proved.  To make the similarity
more salient, we reword \Cref{lem:fo2k} as the equivalent statement:
\begin{lemma}
  If an lm-variety of regular languages \(\cV\) satisfies:
  \[\Delta_2[\reg] \subseteq \cV \subseteq \FO[\arb] \text{ and } \cV \cap \NL \subseteq \Delta_2[<].\]
  then \(\cV = \Delta_2[\reg]\).
\end{lemma}
\begin{proof}
  The second part of the proof of \Cref{thm:strdelta2} shows that
  \(\FO^2[\reg] = \Delta_2[\reg]\), while the same was already known if \(<\) is the only
  available numerical predicate, so we can freely swap \(\FO^2\) for
  \(\Delta_2\) in the statement of \Cref{lem:fo2k}.  We need to show that the
  hypothesis \(\cV \cap \NL \subseteq \Delta_2[<]\) is equivalent to
  \(K \notin \cV\).  For the left-to-right implication, it is known~\cite{kouckypt05} that
  \(K \notin \FO^2[<] = \Delta_2[<]\).  For the right-to-left implication, we use the
  hypothesis that \(\cV \subseteq \FO[\arb]\): \cite[Theorem~9]{cadilhacp21} shows that
  \(\cV \subseteq \FO[\arb]\) and \(K \notin \cV\) implies that
  \(\cV \subseteq \FO^2[\reg]\) and \cite[Theorem~15]{cadilhacp21} asserts that
  \(\FO^2[\reg] \cap \NL \subseteq \mathmbox{\FO^2[<]}\).  Hence
  \(\cV \cap \NL \subseteq \FO^2[<]\) and replacing \(\FO^2\) with \(\Delta_2\) concludes the proof.
\end{proof}

A \emph{positive lm-variety} is defined just as lm-variety, but without
requiring closure under complement.  The statement we need for \(\Sigma_2\) thus reads:
\begin{conjecture}\label{conj:loc}
  If a positive lm-variety of regular languages \(\cV\) satisfies:
  \[\Sigma_2[\reg] \subseteq \cV \subseteq \FO[\arb] \text{ and } \cV \cap \NL \subseteq \Sigma_2[<]\]
  then \(\cV = \FO^2[\reg]\).
\end{conjecture}

If the conjecture held, then using \(\cV = \sta \cap \Reg\) would show the
Straubing Property for \(\sta\).




\subsection{On the fine structure of \(\AC^0\)}\label{sec:conscirc}

For this section, we use notations similar to~\cite{macielpt00} on circuit
complexity (these correspond to the classes \(\text{BC}^0_i\)
therein):\mcnote{Note that Thérien defines these classes as the \emph{Boolean
    closure} of what I wrote.  It's not all that important, I believe, as they
  show results on BC and translate them to AC.}
\begin{itemize}
\item \(\AC^0_i\) is the class of polynomial-size, depth-\(k\) Boolean circuit
  families, where all the circuits in a family have the same kind of output gate (AND
  or OR);\mcnote{This is a bit awkward, but we need it for smooth lower bounds.}
\item \(\widehat{\AC^0_i}\) is defined similarly, the only difference being that
  the input gates of the circuits are allowed to compute any function of at most
  a constant number of positions of the input string.  This class is equivalent
  to \(\Sigma_i[\arb] \cup \Pi_i[\arb]\) (the \(i\) here is the number of quantifiers blocks,
  so that there are \(i-1\) alternations between \(\exists\) and \(\forall\)).
\end{itemize}

The implied hierarchies interleave in a strict way:
\[\AC^0_1 \subsetneq \widehat{\AC^0_1} \subsetneq \AC^0_2 \subsetneq \cdots\]
The strictness of the hierarchy was independently obtained by \cite{macielpt00} and
\cite{caich98}, both relying on previous bounds by
Håstad~\cite{hastad89}, but with very different approaches.  However, the
question of finding \emph{explicit} languages that separate this hierarchy is
still open, to the best of our knowledge.  We make some modest progress towards
this:\mcnote{Can we quickly find languages to separate the lower levels?}
\mcnote{Is it really interesting to have \(K'\)?  We show that \(D_1^{(2)}\) also
  fits that bill below.  I wouldn't want this to look like we're trying to fill
  pages.}
\begin{theorem}
  The language \(K' = K\cdot bc^*b\cdot A^*\) is in \(\AC^0_3 \setminus \widehat{\AC^0_2}\).
\end{theorem}
\begin{proof}
  \Paragraph{Upper bound.} It is easily seen that a word is in the language
  \(K\) iff:
  \begin{itemize}
  \item Between every two \(a\)'s there is a \(b\), and vice versa;
  \item The first (\resp last) nonneutral letter of the word is \(a\) (\resp \(b\)).
  \end{itemize}
  Each of these statements can be written as an AND of ORs; for
  the first one, it is easier seen on the complement: we have an OR gate that
  selects two positions \(p_1, p_2\) and checks that there is an \(a\) at both
  positions and only \(c\)'s in between.  Thus \(K\) can be written as an AND of
  these, and so \(K\) has a circuit of depth exactly 2.

  To build a circuit for \(K'\), we start with an OR gate that selects two
  positions \(p_1, p_2\), and checks with an AND gate that they both contain a
  \(b\) and that only \(c\)'s appear in between.  We add as input to that AND gate
  the inputs of the AND gate for \(K\) where the positions considered are
  restricted to be smaller than \(p_1\).  This thus correctly checks that the
  prefix up to \(p_1\) is in \(K\), and that it is followed by a word starting with
  \(bc^*b\).

  \paragraph{Lower bound.}  If \(K' \in \widehat{\AC^0_2}\), then \(K'\) is either in
  \(\Sigma_2[\arb]\) or \(\Pi_2[\arb]\).

  Assume \(K' \in \Sigma_2[\arb]\), then \(K' \in \sto\) by \Cref{thm:main}.  We show that
  \(K' \notin \sto\) using the equations provided by \Cref{thm:eqs} with the wording
  appearing after the Theorem.  First, the word \((ab)^2\) is mapped to an
  idempotent: if \((ab)^2\) appears in a word, we can repeat it any number of
  times without changing membership to \(K'\).  Also, the word \(ba\) appears as a
  subword of \((ab)^2\).  However \((ab)^6\cdot bb \in K\), but \((ab)^2(ba)(ab)^2\cdot bb \notin
  K'\), showing that \(K' \notin \sto\).

  If we assume that \(K' \in \Pi_2[\arb]\), we have to show that the complement of
  \(K'\) is not in \(\sto\).  This time, we pick \((ab)^3\) as the idempotent, and
  \(bba\) as the subword.  Then \((ab)^9 \notin K'\) but \((ab)^3(bba)(ab)^3 \in K'\),
  hence the complement of \(K'\) is not in \sto.
\end{proof}

The language \(K\) itself appears very often in the literature pertaining to the
fine separation of small circuit classes~\cite{kouckypt05,brzozowskik78}.  This is no surprise:
The language is the first of the family of \emph{bounded-depth Dyck languages}.
These are the well-parenthesized expressions that nest no more than a fixed value:
\begin{align*}
  D_1^{(1)} & = K  = (ac^*b+c)^*,\\
  D_1^{(i)} & = (aD_1^{(i-1)}b+c)^*.\\
\end{align*}
(Here, \(a\) can be interpreted as ``opening parenthesis'' and \(b\) as ``closing.'')

Saliently, these languages separate the class \(\Sigma_i[<]\) from \(\Sigma_{i+1}[<]\), with
\(D_1^{(i)}\) belonging to the latter but not the former~\cite{brzozowskik78}.  It is open
whether these languages also separate the \(\Sigma_i[\arb]\) hierarchy, and thus the
\(\AC^0_i\) one.  We can show that:
\begin{theorem}
  \(D_1^{(2)} \in \AC^0_3 \setminus \widehat{\AC^0_2.}\)
\end{theorem}
\begin{proof}
  \Paragraph{Upper bound.} It can be shown~\cite[Lemma~4\(^+\)]{brzozowskik78} that:
  \[D_1^{(2)} = KbA^* \cup A^*bc^*bKbA^* \cup A^*aK \cup A^*aKac^*aA^*.\]%
  We can use an OR gate to select the positions of the letters \(a, b\) mentioned
  in that expression, and then either use an AND gate to verify that only
  \(c\)'s appear between them, or use the AND of ORs circuit for \(K\) from the
  previous proof to check that a word of \(K\) appears between two positions.

  \Paragraph{Lower bound.}  We again use the wording appearing after
  \Cref{thm:eqs}.  As in the previous proof, we need to show that neither
  \(D_1^{(2)}\) nor its complement are in \(\sto\).  For the language itself, we
  pick \((ab)^3\) as the word mapping to an idempotent and \(bba\) as the subword.
  We indeed have that \((ab)^9 \in D_1^{(2)}\), but
  \((ab)^3(bba)(ab)^3 \notin D_1^{(2)}\).  For the complement, we pick
  \((ab)^2\) as the word mapping to an idempotent and \(aab\) as the subword; they
  satisfy \((ab)^6\cdot b \notin D_1^{(2)}\) but \((ab)^2(aab)(ab)^2\cdot b \in D_1^{(2)}\).
\end{proof}

\section{Conclusion}\label{sec:concl}

We have shown the Neutral Straubing Property for \(\Sigma_2\):
\[\Sigma_2[\arb] \cap \Reg \cap \NL = \Sigma_2[<] \cap \NL.\]
To do so, we developed a new lower bound technique against circuits of depth 3
that relies on the \emph{entailment} relation of a language.  This relation
indicates how dependent the positions within a language are on one another.  We
believe that this relation may be exploited to show Straubing Properties at
higher levels of the \(\Sigma_i\) hierarchy.

Dropping the neutral-letter restriction from our main result is an interesting
task.  Although it would not imply a much stronger statement in terms of
circuits, it is still a stain on the clean statement that is the Straubing
Property.  We note that \Cref{conj:loc} is implied by the so-called
\emph{locality} property of the algebraic counterpart of \(\Sigma_2[<]\); showing
locality is a notoriously hard problem in algebraic language theory~(see, e.g.,
\cite{almeida96, tilson87, straubing85, straubing15}).

Showing Straubing Properties remains a challenging and wide open problem.  If we
are to follow our approach for \(\Sigma_i\), \(i \geq 3\), it requires in particular a
decidable characterization of the form of \Cref{thm:eqs}.  This is usually
provided by an equational characterization of the class, and although the
general shape of the equations for each \(\Sigma_i[<]\) is
known~\cite[Theorem~6.2]{placez19}, they do not readily imply decidability; in
fact, the decidability of \(\Sigma_i[<]\) for \(i \geq 5\) is open~\cite{place18}.

An outstanding example of the connection between the Straubing Property of a
circuit class and its computational power is given
in~\cite{kouckypt05,cadilhacp21}: \(\FO^2[\arb]\) has the Straubing Property if
and only if addition cannot be computed with a linear number of gates.  This
latter question, pertaining to the precise complexity of addition, was asked, in
particular, by Furst, Saxe, and Sipser~\cite[Section~5]{furstss84}.  Note that
even though \(\Delta_2[\reg]=\FO^2[\reg]\), this is not known to hold for the set of
arbitrary numerical predicates.

\begin{acks}
  We wish to thank Nikhil Balaji and Sébastien Tavenas.  The last author
  acknowledges financial support by the DFG grant ZE 1235/2-1.
\end{acks}

\bibliographystyle{ACM-Reference-Format}
\bibliography{bib}

\end{document}
\endinput
